\documentclass[12pt]{article}
\usepackage[textwidth=500pt,textheight=650pt,centering]{geometry}
\usepackage{amsmath,amssymb,amsthm}
\usepackage{hyperref}
\usepackage{graphicx}
\usepackage[usenames]{color}

\title{Fluctuations in a kinetic transport model for quantum friction}

\date{March 14, 2014}

\author{Roland Bauerschmidt\thanks{Present address:
    School of Mathematics, Institute for Advanced Study,
    Einstein Drive, Princeton, NJ 08540 USA. E-mail: {\tt brt@math.ias.edu}.},\;
  Wojciech de Roeck\thanks{Present address:
    Theoretical Physics Section, Celestijnenlaan 200d - box 2415,
    3001 Leuven, Belgium. E-mail: {\tt wojciech.deroeck@fys.kuleuven.be}.}\;
  and J\"urg Fr\"ohlich\thanks{Institut f\"ur Theoretische Physik,
    ETH Z\"urich, 8050 Z\"urich, Switzerland. Present address: School of
    Mathematics, Institute for Advanced Study, Einstein Drive, Princeton,
    NJ 08540 USA. E-mail: {\tt juerg@itp.phys.ethz.ch}.}}

\newcommand{\N}{\mathbb{N}}

\newcommand{\R}{\mathbb{R}}

\renewcommand{\Pr}{\mathbb{P}}
\newcommand{\E}{\mathbb{E}}

\newcommand{\dd}[2]{\frac{d #1}{d #2}}
\newcommand{\half}{{\frac{1}{2}}}

\numberwithin{equation}{section}
\newtheorem{thm}{Theorem}[section]

\newtheorem{lem}[thm]{Lemma}
\newtheorem{prop}[thm]{Proposition}

\begin{document}
\maketitle

\begin{abstract}
  We consider a linear Boltzmann equation that arises in a model for
  quantum friction. It describes a particle that is slowed
  down by the emission of bosons.  We study the stochastic process
  generated by this Boltzmann equation and we show convergence of its
  spatial trajectory to a multiple of Brownian motion with exponential
  scaling.  The asymptotic position of the particle is finite in mean,
  even though its absolute value is typically infinite. This is
  contrasted to an approximation that neglects the influence of
  fluctuations, where the mean asymptotic position is infinite.
\end{abstract}

\nocite{Bau09}
\section{Introduction}

\subsection{Motivation} 
One of the interesting themes in non-equilibrium physics is the role
of fluctuations. For example, ratchets (also known as molecular
motors) function because they can convert fluctuations into a directed
motion, something that would be impossible in an equilibrium (detailed
balance) process.  The model described in this paper offers another
very simple illustration of the fundamental role of fluctuations.  It
describes a quantum-mechanical particle (hereafter called ``tracer
particle'') interacting with the atoms in an ideal Bose gas at zero
temperature exhibiting Bose--Einstein condensation. The effective
dynamics of the tracer particle is described by a linear Boltzmann
equation, i.e., by a stochastic process.  Lately, the analysis of such
transport equations has received considerable attention; see
\cite{MR2588245,MR2664334,MR2763032}, and \cite{MR1295030} for a
general introduction.  They can appear as linearizations of nonlinear
Boltzmann equations close to equilibrium, or as models of a single
entity interacting with a medium in thermal equilibrium, as in our
example.  Below, we describe the stochastic process corresponding to
our linear Boltzmann equation and we discuss some of its properties,
postponing a sketch of its derivation from a physical model to
Section~1.4. (Actually, the physics origin of the equation studied in
this note has no importance for our analysis.)

\subsection{Jump process} \label{sec: jump process} Let $(X_t,K_t) \in
\R^3 \times \R^3 $ stand for the position/momentum of the tracer
particle at time $t$.  The momentum $K_t$ is a Markov jump process
defined by a \emph{jump kernel} $q(\cdot,\cdot)$ as follows: For each
momentum $k \in \R^3$, $q(k,\cdot)$ is a finite measure on
$\mathbb{R}^{3}$; its total \emph{weight}, $\Sigma(k):= q(k,\R^3)$, is
the rate at which a jump in momentum space occurs, starting from
momentum $k$, and $q(k, B)/\Sigma(k)$ is the probability for such a
jump to land in the set $B \subset \R^3$.  Formally, the jump kernel
$q$ is given by
\begin{equation} \label{eq:q1}
  q(k, dk') = w  \, \delta(\varepsilon(k')- \varepsilon(k)-\omega(k'-k)) \; dk',
\end{equation}
where the dispersion laws
\begin{equation} \label{eq:eps-omega}
  \varepsilon(k) = \frac{1}{2m} |k|^2, \quad \omega(q) = \frac{1}{2M} |q|^2, \quad (k,q \in \R^3),
\end{equation}
correspond to the kinetic energies, as functions of the momenta, of
the non-relativistic tracer particle and the Bose atoms, respectively,
$m$ and $M$ are their masses, and $w$ is the effective coupling
strength of the interaction between the tracer particle and an atom in
the Bose gas. In what follows, we will choose units such that $w=1$.
Physically, the delta function in \eqref{eq:q1} expresses conservation
of energy in a momentum-preserving process where the tracer particle
excites an atom in the Bose--Einstein condensate to a state of
momentum $k-k'$ and energy $\omega(k-k')$. (This can be viewed as the
emission of a sound wave into the Bose gas.)

The position $X_t$ evolves according to
\begin{equation}
\frac{d X_t}{dt} = \frac{K_t}{m}.
\end{equation}
Note that $X_t$ is a random variable, because $K_t$ is random.
Obviously, since, at zero temperature, the particle can only emit (but
not absorb) sound waves, thereby lowering its energy, we expect that
$K_t \to 0$, i.e., that the particle experiences friction.  Our main
interest is in the behavior of the position $X_t$, including its
asymptotics as $t \rightarrow \infty$.  The following scaling law (cf.\
\eqref{eq:q} below) is crucial for our result:
\begin{equation} \label{eq: scaling reation}
  \lambda q(\lambda k,\lambda B)  = q(k, B), \qquad \text{for any $\lambda >0$.}
\end{equation}
When combined with rotation invariance of the kernel, it leads to the formal relation
\begin{equation} \label{eq: expectation difference}
  \E(K_{t+d t}-K_t| K_t) = - \eta | K_t |  K_t d t, 
\end{equation}
where $\E( \cdot | K_t) $ is the expectation of the momentum process conditioned on $K_t$,
and $\eta >0$ is some friction coefficient depending on $m,M$.

\subsection{Mean field approximation vs.\ fluctuations}
A first guess concerning the behavior of our process is that one may
\emph{neglect} random fluctuations. This amounts to omitting the
expectation $\E$ in \eqref{eq: expectation difference} and pretending
that $K_{t+d t}-K_t$ is \textit{nonrandom} and hence that $K_t$
satisfies the differential equation
\begin{equation} \label{eq: diff equation}
 \dd{K_t }{t}    =  - { \eta | K_t |  K_t }, \qquad \text{with} \text{   }K_{t=0} =K_0,
\end{equation}
whose solution is given by 
\begin{equation}
  K_t    =   (\eta t+1/| K_0 |)^{-1}    \frac{K_0}{| K_0 |}.
\end{equation}
As a consequence, we find that the position $X_t$ diverges logarithmically
\begin{equation}
  |X_t |  \sim \text{const} \, \log t, \qquad \text{as $t\to \infty$.}
\end{equation}
However, after including fluctuations, the behavior of $X_{t}$ turns
out to be different: $\E( X_t )$ remains finite, while typical
trajectories, $(X_t)_{0\leq t < \infty}$, escape to infinity in random
directions, as $t\rightarrow \infty$.  In fact, the scaling relation
\eqref{eq: scaling reation} implies that the typical time elapsing
before the momentum jumps again, starting from a momentum $K_t$, is of
order $| K_t |^{-1}$, hence the typical distance the tracer particle
travels between two successive momentum jumps remains bounded, as
$K_t \to 0$. Rescaling time, we can alternatively think of the path
$(X_t)_{0\leq t \leq T}$ as the path of a particle traveling with unit
speed and randomly changing its direction of motion at a uniform rate,
which is stopped at a time of order $\log T$.  In particular, after
appropriate rescaling, $X_t$ can be expected to converge to a multiple
of Brownian motion.  This is, in fact, our main result, which is
stated precisely in Theorem \ref{thm:convergence}, below.

\subsection{Derivation from first principles}

The physical system motivating our study of the particular model
described above consists of a non-re\-la\-ti\-vist\-ic, quantum tracer
particle of mass $m$ interacting with non-relativistic atoms of mass
$M$ in an ideal Bose gas at zero temperature exhibiting Bose--Einstein
condensation. The Bose--Einstein condensate serves as a reservoir of
atoms that, through soft collisions, the tracer particle can lift to
traveling wave states of non-vanishing momentum and positive kinetic
energy. This corresponds to an emission of \v{C}erenkov radiation of
sound waves into the Bose gas --- similarly to a well known phenomenon
of light emission observed when a charged particle moves through an
optically dense medium at a speed larger than the speed of light in
the medium; see, e.g., \cite{MR0436782}.
One may thus expect that the tracer particle experiences friction and
slows down until, asymptotically, it comes to rest.

The dynamics of the physical system described above is determined by the many-body 
Schr\"o\-din\-ger Hamiltonian
\begin{equation} \label{eq:H_I}
H^{(N)}= -\frac{\Delta_X}{2m} + \sum_{n=1}^{N}\left\lbrace -\frac{\Delta_{x_n}}{2M} + gW(x_{n}-X) \right\rbrace,
\end{equation}
where $\Delta$ denotes the Laplacian (with suitable boundary
conditions imposed at the boundary of a cube in $\mathbb{R}^{3}$ to
which the system is confined), $gW(x_{n}-X)$ is the two-body
interaction potential when the tracer particle is at position $X$ and
the $n^{th}$ atom of the Bose gas at position $x_n$, $n=1, \dots ,N$,
and $g$ is a coupling constant. We assume that $W$ is bounded and of
rapid decrease at $\infty$, and that the density of the gas is
positive and kept constant, as the thermodynamic limit is
approached. The operator $H^{(N)}$ introduced in \eqref{eq:H_I} is
self-adjoint and bounded below on the usual $L^{2}-$space of orbital
$(N+1)-$particle wave functions and generates the time evolution of
the system, for all $N<\infty.$

We are interested in studying the ``effective time evolution" of the
tracer particle when the Bose gas is initially in a state of thermal
equilibrium at some temperature $\beta^{-1}\geq0$. This evolution is
obtained by taking the expectation over the degrees of freedom of the
Bose gas corresponding to an equilibrium state, conditioned on the
state of the tracer particle, and taking $N\to \infty$ to eliminate
finite-size effects.  In \cite{MR1901513}, Erd\H{o}s has shown that,
for positive temperatures, $\beta^{-1} > 0$, and after rescaling $x
\mapsto g^{-2}x$, $t \mapsto g^{-2}t$ (``weak kinetic scaling''), the
limit $g\rightarrow 0$ of the effective dynamics of the tracer
particle exists. Moreover, in this so-called \emph{kinetic limit}, the
evolution of the Wigner distribution of the state of the tracer
particle is given by a solution of the linear Boltzmann equation whose
collision operator is given by the kernel
\begin{multline}
  q(k, dk') = |\widehat{W}(k-k')|^2  \big[
    (N(k-k') +1 ) \delta(\varepsilon(k)-\varepsilon(k')-\omega(k-k')) d k'+ \\
    N(k-k') \delta(\varepsilon(k')-\varepsilon(k)-\omega(k'-k)) d k'
  \big].
\end{multline}
Here $\widehat{W} $ is the Fourier transform of the potential $W$, and $N(q)$ is given by
\begin{equation} \label{eq:N}
  N(q) = \frac{e^{-\beta \omega(q)}}{1-e^{-\beta\omega(q)}}
  = n(\beta\omega(q)), \quad n(x) = \frac{1}{e^x - 1},
\end{equation}
with $n(x)$ the Bose--Einstein distribution function.
The results in \cite{MR1901513} only hold for positive temperatures, $\beta^{-1}>0$.
However, one may formally pass to zero temperature and finds that the collision kernel then reduces to
\begin{equation}
  q(k, dk') = |\widehat{W} (k-k')|^2 \delta(\varepsilon(k)-\varepsilon(k')-\omega(k-k')) dk'.
\end{equation}
Furthermore, since we are interested in phenomena at small momenta,
only the behavior of $\widehat{W}$ near $0$ matters. As $W$ has been
assumed to have rapid decay at $\infty$, $\widehat{W}(k)$ is smooth
near $k=0$, and, to simplify matters, it is reasonable to replace the
function $\widehat{W}$ by the constant $w=|\widehat{W}(0)|^2$, which
is henceforth chosen to be equal to unity. This then yields the jump
kernel introduced in \eqref{eq:q1}.  (The replacement of
$\widehat{W}$ by a constant, $\widehat{W}(0)$, is made in order to
avoid uninteresting technical complications. We expect, however, that
all our results hold for an arbitrary $\widehat{W}$ that is continuous
at $0$.)

As our discussion may have made plausible, a rigorous understanding of
friction in the realm of unitary quantum dynamics or classical
Hamiltonian mechanics still remains an excellent mathematical
challenge.  Some results in this direction may be found in
\cite{MR1924366,MR2212220,MR2858064}.

\section{Result}

To begin with, we carefully introduce the linear Boltzmann equation
and the associated stochastic process that have been described in
Section \ref{sec: jump process}.

On the space of functions $f \in C^1(\R_x^3 \times \R_k^3)$ (we will
sometimes write $\R_x^3$, instead of $\R^3$, to emphasize that the
variable in $\R_x^3 = \R^3$ is denoted by $x$), we define a linear
operator $L$ by setting
\begin{equation} \label{eq:L}
  Lf =
  \frac{k}{m} \cdot \nabla_x f
  + M f,
\end{equation}
where $M$ is the linear collision operator
\begin{equation} \label{eq:M}
  Mf(k) = \int_{\R_k^3} (f(k')-f(k)) \delta(\varepsilon(k)-\varepsilon(k')-\omega(k-k')) \; dk'.
\end{equation}
The $\delta$-function in the integrand enforces conservation of energy.
It is the composition of an ordinary $\delta$-function with the function
\begin{equation} \label{eq:def-F}
  F(k') = F_{k}(k') = \varepsilon(k)-\varepsilon(k')-\omega(k-k')
\end{equation}
and is defined, more precisely, by
\begin{equation} \label{eq:def-delta}
  \int f(k') \delta(F(k')) \; dk' = \int_{F^{-1}(0)} f(k') \, |\nabla F(k')|^{-1} \; H(dk'),
\end{equation}
where
$H$ is the unnormalized surface (Hausdorff) measure on $F^{-1}(0)$.

It is a standard fact that the operator $L$ introduced in \eqref{eq:L}
generates a Markovian semigroup.
We adopt the usual probabilistic point of view in which the semigroup is thought to act
on observables (``Heisenberg picture''), so that the expectation of a test function 
$f: \R_x^d \times \R_k^d \to \R$ at time $t$
is given by $\E(f_t(X_0,K_0))$, where $f_t$ is the unique solution to the forward equation
\begin{equation}
  f_0 = f, \qquad \partial_t f_t = L f_t, \quad (t>0),
\end{equation}
and $\E(\cdot)$ denotes the expectation corresponding to the random variable $(X_0,K_0)$.

In the following, $S^2 \subset \R^3$ denotes the unit sphere, and $U$
is a uniform random variable on $S^2$. Let $(X_t,K_t)$ be a
realization of the Markov process with generator \eqref{eq:L}, started
from a deterministic point $(X_0,K_0)=(x,k)$ (i.e., $\E(f(X_t, K_t)) =
f_t(x,k)$), and let $(B_s)$ be standard \textit{Brownian motion} on
$\R_x^3$. The main result in this note is the following theorem.

\begin{thm} \label{thm:convergence}
  There are explicit positive constants $\theta<1$ and $\sigma$ such
  that, for an arbitrary choice of initial conditions $(X_0,K_0) =
  (x,k) \in \R^3 \times \R^3$,
  \begin{equation} \label{eq:convergence}
    \left(\frac{1}{\sqrt{n}}X_{\theta^{-ns}}\right)_{s} \to (\sigma B_{s})_{s},
    \quad \text{as $n\to\infty$,}
  \end{equation}
  in distribution,
  in the topology of uniform convergence on bounded intervals of the variable $s$. 
\end{thm}

This result shows in particular that, no matter what the initial
conditions are, the position of the tracer particle is not confined
to any bounded region in $\R^3$, but that the distance of its
position from the initial position only grows very slowly. Its
motion is \textit{diffusive} on an exponential time scale.
The mean position of the particle remains, however, finite, as time
tends to $\infty$, because the momentum of the particle rapidly
loses memory of its initial values. More precisely, from the proof
of Theorem~\ref{thm:convergence}, we will infer the following
result.
\begin{thm} \label{thm:position}
  \begin{equation} \label{eq:Kasymp}
    |K_{t}| = |K_0|t^{-1+o(1)} \quad \text{almost surely,}
  \end{equation}
  and
  \begin{equation} \label{eq:XEinf}
    \lim_{t\to \infty} \E|X_t|=\infty \qquad \text{but} \qquad \sup_{t\geq 0} |\E X_t | <\infty.
  \end{equation}
\end{thm}

The proofs of these theorems are probabilistic and rely on a
decomposition of the continuous-time Markov process $(X_t, K_t)$ into
its discrete-time skeleton process and on a process of jump times.
For the jump times, we prove a law of large numbers on an exponential
scale, and for the skeleton process we derive a functional central
limit theorem.  The theorems stated above follow from these two
results by standard arguments.

\medskip

The constants $\theta \in (0,1)$ and $\sigma > 0$ in Theorem~\ref{thm:convergence}
are explicit functions of $m$ and $M$ and are given by
\begin{equation} \label{eq:thetasigmadef1}
  \log\theta
  = \half \left(\frac{(1-2 a)^2 \log |1-2 a|}{2 (a-1) a}-1\right),
  \qquad
  \sigma^2 = \frac{2}{3(1-b)} \frac{(m+M)^4}{16\pi^2 m^4M^4},
\end{equation}
where
\begin{equation} \label{eq:abdef1}
  a := \frac{m}{m+M},
  \qquad
  b =  \half \frac{|1-2 a| \left(2 a^2+a-1\right)-3 a+1}{3 (a-1) a^2}
  .
\end{equation}

In our proofs, the constants $b$ and $\theta$ arise from the equations
\begin{align} \label{eq:constU}
  by &=
  \frac{1}{4\pi} \int_{S^2} \frac{ay+(1-a)z}{|ay+(1-a)z|} \; H(dz),
  \\
  \log\theta &=
  \frac{1}{4\pi} \int_{S^2} \log(|ay+(1-a)z|) \; H(dz)
  ,
\end{align}
for any $y \in S^2$, as can be verified by using polar coordinates on $S^2$.

\section{Proof of Theorem~\ref{thm:convergence}}

\subsection{Stochastic decomposition}

Recall that $a = m/(m+M)$. A straightforward calculation shows that $F_k(k') = 0$, with $F$ given by \eqref{eq:def-F}, is equivalent to 
\begin{equation}
  k' \in ak+(1-a)|k|S^{2} \subset \R^3;
\end{equation}
see Figure~\ref{fig:surfaces}.
\begin{figure}[h]
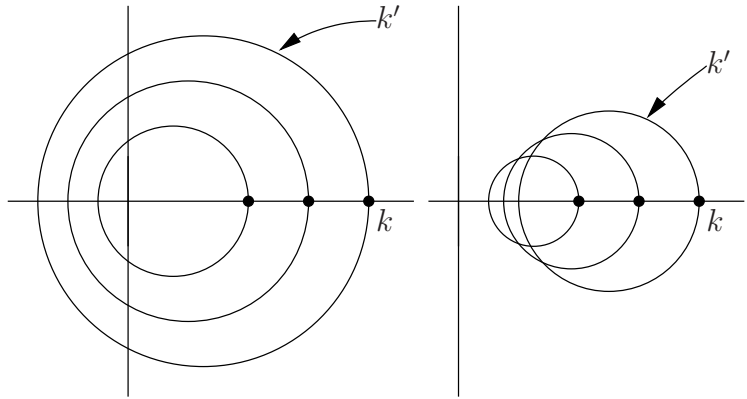

  \begin{center}
    \input{a_leq_0.5.pspdftex}
    \input{a_geq_0.5.pspdftex}
  \end{center}
  \caption{The surfaces $k'\in ak+(1-a)|k|S^{2}$ for $a < \half$ (left) and $a> \half$ (right).}
  \label{fig:surfaces}
\end{figure}
This observation and the equation
\begin{equation}
  \nabla_{k'} F_k (k')
  = - \nabla \varepsilon(k') + \nabla \omega(k-k')
  = \frac{1}{M} \left(k-\frac{k'}{a}\right)
\end{equation}
imply that, for $F_{k}(k')=0$,
\begin{equation}
  |\nabla_{k'} F_k (k')|
  = \frac{1-a}{aM} |k| = \frac{|k|}{m}
  .
\end{equation}
Let $\lbrace u(k, dk')\rbrace$ be the family of unnormalized surface
(Hausdorff) measures on the surfaces given by $ak+(1-a)|k|S^{2}$, for
$k\in \mathbb{R}^{3}$. They are given by affine transformations of the
uniform measure on the unit sphere.  Expressed in terms of these
measures, the kernel of the collision operator \eqref{eq:M} is given
by
\begin{equation} \label{eq:q}
  q(k, dk') = \frac{m}{|k|} \; u(k, dk').
\end{equation}
For a given $k$, the \emph{scattering rate} $\Sigma(k)$ is the total weight of $q(k,\cdot)$:
\begin{equation} \label{eq:Sigma}
  \Sigma(k) = q(k, \R^3) 
  = \frac{m}{|k|}u(k,\R^3) = m(1-a)^2 |S^2| |k| = 4\pi \frac{M^2m}{(m+M)^2} |k|
  .
\end{equation}
The normalized transition kernel,
\begin{equation} \label{eq:p}
  p(k, dk') = \Sigma(k)^{-1} q(k, dk'),
\end{equation}
is simply the uniform probability measure on the sphere $ak+(1-a)|k|S^2$.

Let $(\Omega, \Sigma, \Pr)$ be a probability space $\Omega$ with sigma-algebra $\Sigma$
and probability measure $\Pr$ on which the following random variables are defined:
$(K_n)$ is a Markov chain on $\R^3$ with $K_0=k$ and transition probabilities given in \eqref{eq:p},
and $(\lambda_n)$ is a sequence of independent random variables with exponential
probability distribution with mean 1.
We also let $X_0$ be the deterministic random variable $X_0=x$.
Let $\Sigma_n$ be the sigma-algebra generated by $(\lambda_j,K_j)_{j\leq n}$.
Let
\begin{equation}
  T_n = \sum_{j=0}^{n-1} \frac{\lambda_j}{\Sigma(K_j)},
\end{equation}
and let $N_t$ be its right-continuous inverse. Then $T_n$ is the time at which the $n$th jump occurs,
and $N_t$ is the number of jumps up to time $<t$. We set 
\begin{equation} \label{eq:KtXt}
  K_t = K_{N_t}, \quad X_t = X_0 + \int_0^t \frac{K_t}{m} \; dt,
\end{equation}
where, with some abuse of notation, $K$ is an abbreviation for both
the original continuous-time process $K_t$ and the discrete time
process $K_n$; (when using a subscript $t \in \R_+$ we mean the former
process, while, for a subscript $n \in \N$, the latter process is
meant).

\begin{prop} \label{prop:markovboltzmann}
  $(X_t,K_t)_{t\geq 0}$ is a strong Markov process. 
  Let $(\Gamma_t)$ be its semigroup defined by $(\Gamma_tf)(x,k) = \E^{x,k}(f(X_t, K_t))$,
  where $\E^{x,k}$ is the expectation with $(X_0,K_0)=(x,k)$.
  Then the generator of $(\Gamma_t)_{t\geq 0}$ is the operator $L$ given by \eqref{eq:L}.
\end{prop}

The proof is a standard argument, presented for completeness in
Appendix~\ref{sec:pfmarkovboltzmann}.
We study the skeleton process $(K_n)$ in terms of its polar decomposition
\begin{equation} \label{e:def-YR}
  Y_n = K_n/|K_n|, \quad R_n = |K_n|.
\end{equation}
The main observation is that $(Y_n)$ is a Markov chain on $S^2$, and
\begin{align} \label{eq:Xtsum}
  X_{t}
  &= X_0  + \frac{1}{m} \sum_{j=0}^{N_t-1} (T_{j+1}-T_{j}) K_j  + \frac{1}{m} (t-T_{N_t}) K_{N_t} \nonumber\\
  &= X_0 + \frac{(m+M)^2}{4\pi m^2M^2} \sum_{j=0}^{N_t-1} \lambda_j Y_j + \frac{1}{m} (t-T_{N_t}) K_{N_t}
  .
\end{align}

Theorem~\ref{thm:convergence} is a consequence of the following two propositions for the sum over $j$ of
$\lambda_j Y_j$ and for the number of jumps $N_t$.
In the statements of the propositions, $\theta,b \in (0,1)$ are the
constants defined in \eqref{eq:thetasigmadef1}--\eqref{eq:abdef1}, and
$D([0,s_0],\R^3)$ is the Skorohod space of right-continuous functions
with left-limits (\emph{c\`adl\`ag functions}) $[0,s_0] \to \R^3$,
endowed with the Skorohod topology; see, e.g., \cite{MR1700749}.

\begin{prop} \label{prop:clt}
  For $s \in [0,s_0]$, where $s_0>0$ is arbitrary,
  \begin{equation} \label{eq:clt}
    \left(\frac{1}{\sqrt{n}} \sum_{j=0}^{[ns]} \lambda_j Y_j\right)_{s} \stackrel{D}{\to} \left(\sqrt{\frac{2}{3(1-b)}} B_s\right)_{s},
    \quad \text{as $n\to\infty$,}
  \end{equation}  
  in distribution in the Skorohod space $D([0,s_0],\R^3)$.
\end{prop}

\begin{prop} \label{prop:lln}
  Uniformly for $s \in [0,s_0]$, where $s_0>0$ is arbitrary,
  \begin{equation}
    N_{\theta^{-ns}}/n \to s, \quad \text{as $n\to\infty$,} \quad
    \text{almost surely}.
\end{equation}
\end{prop}

The proofs of the propositions are deferred to Sections~\ref{sec:pflln}--\ref{sec:pfclt}.
The left-hand side of \eqref{eq:clt} is an additive functional of an exponentially
mixing Markov process, and there are many approaches to proving such a functional
central limit theorem; we use the martingale method.
Given these two propositions, Theorem~\ref{thm:convergence} is proved as follows.

\begin{proof}[Proof of Theorem~\ref{thm:convergence}]
  Fix $s_0>0$.
  Since the topology of uniform convergence on $C([0,s_0],\R^3)$ coindices with the Skorohod topology restricted to 
  continuous functions, and since $s \mapsto X_{\theta^{-ns}}$ is continuous almost surely, which is evident from \eqref{eq:KtXt},
  it suffices to show that $\frac{1}{n}X_{\theta^{-ns}} \to \sigma B_s$ as processes in $D([0,s_0],\R^3)$.
  Let
  \begin{equation}
    Z_n(s) = \frac{1}{\sqrt{n}} \frac{(m+M)^2}{4\pi m^2M^2} \sum_{j=0}^{[ns]} \lambda_jY_j,
    \quad
    \Phi_n(s) = \frac{N_{\theta^{-ns}}-1}{n}
    .
  \end{equation}
  We claim that $Z_n \circ \Phi_n \to \sigma B$ in distribution in $D([0,s_0],\R^3)$.
  This claim can be shown following \cite[Section~17]{MR1700749}:
  Let $D_0$ be the subspace of $D([0,s_0],[0,\infty))$ of nondecreasing functions
  (endowed with the relative topology).
  Let $\Phi(s) = s$.
  Then $\Phi_n,\Phi \in D_0$,
  and Proposition~\ref{prop:lln} with the fact that the Skorohod topology is weaker than
  the uniform topology implies $\Phi_n \to \Phi$ almost surely as elements of $D_0$.
  By \cite[Theorem~4.4]{MR1700749}, with $\Phi_n \to \Phi$ a.s.\ in $D_0$ and
  Proposition~\ref{prop:clt}, it follows that, in distribution in $D\times D_0$,
  \begin{equation}
    (Z_n,\Phi_n) \to (\sigma B, \Phi).
  \end{equation}
  Since $B$ and $\Phi$ are continuous, this implies that $Z_n \circ \Phi_n \to \sigma B$;
  see \cite[Section~17]{MR1700749}.

  It remains to argue that the last term in \eqref{eq:Xtsum} is negligible.
  By \cite[Theorem~4.1]{MR1700749}, it suffices to show that
  \begin{equation} \label{e:lastterm}
    \frac{1}{\sqrt{n}} \sup_{s\in[0,s_0]} |(s-T_{N_{\theta^{-ns}}})K_{N_{\theta^{-ns}}}|
    \leq  \frac{1}{\sqrt{n}} \max_{j \leq N_{\theta^{-ns_0}}} \lambda_j \to 0
    \quad \text{(in probability)}.
  \end{equation}
  Since
  $\E(\max\{\lambda_1, \dots, \lambda_k\}) = \sum_{i=1}^k \frac{1}{i} = O(\log k)$
  and by Proposition~\ref{prop:lln},
  \begin{align}
    \Pr\{\max_{j \leq N_{\theta^{-ns_0}}} \lambda_j \geq \varepsilon \sqrt{n}\}
    &\leq
    \Pr\{\max_{j \leq n^2} \lambda_j \geq \varepsilon \sqrt{n}\}
    + \Pr\{N_{\theta^{-ns_0}} \geq n^2\} = o(1),
  \end{align}
  showing \eqref{e:lastterm}, as claimed.
\end{proof}

It remains to prove Propositions~\ref{prop:clt}--\ref{prop:lln}, and
to prove Theorem~\ref{thm:position}, a task addressed in the rest of the paper.

\subsection{Proof of Proposition~\ref{prop:lln}}
\label{sec:pflln}

Before giving the proof of Proposition~\ref{prop:lln},
we briefly sketch its main idea.
It will be shown that $\frac{1}{n} \log \Sigma(R_n)$ concentrates at $\log \theta$, as $n\to\infty$,
so that $\Sigma(R_n) \approx \theta^n$. Since $\theta^{-n}$ grows exponentially, as $n\to\infty$,
the last term in
\begin{equation}
    \sum_{j=0}^n \lambda_j\Sigma(R_j)^{-1}
\end{equation}
is dominant. In particular, a lower bound on this sum can be obtained by dropping all terms except
for the last one. On the other hand, using that $\Sigma(R_j)$ is decreasing, an upper bound
is obtained by replacing $\Sigma(R_j)^{-1}$ by $\Sigma(R_n)^{-1}$ for $j = 0, \dots, n-1$.
Thus, on the exponential scale, the above sum can be approximated well by $\theta^{-n}$.
(Very crude estimates suffice here, due to the exponential scaling.)
The claim can be deduced from such considerations, recalling that, by definition,
\begin{equation} \label{eq:Nthetat}
  N_{\theta^{-s}} = \inf \left\{ n \geq 0: \sum_{j=0}^n \lambda_j\Sigma(R_j)^{-1} \geq \theta^{-s}\right\}
  .
\end{equation}

We now enrich the above sketch to a proof.
Since $p(k,\cdot)$ is the uniform probability measure on $ak+(1-a)|k|S^2$,
it follows that $K_{j+1} = aK_{j}+(1-a)|K_{j}|U_j$,
where $U_j$ is a uniform random variable on $S^2$. Thus $R_{j+1} = R_{j}|aY_j+(1-a)U_j|$.
In distribution, $|aY_j+(1-a)U_j|$ is equal to $|ay+(1-a)U_j|$, with any $y \in S^2$,
and therefore $D_j = \log R_{j} - \log R_{j-1}$ are i.i.d.\ random variables with
distribution $\gamma_*u$ where, for an arbitrarily fixed $y \in S^2$,
\begin{equation} \label{eq: definition gamma}
  \gamma(z) 
  = \log |ay+(1-a)z|, \quad (z \in S^{2}),
\end{equation}
$u$ is the uniform probability measure on $S^{2}$, and $\gamma_*u$ is its push-forward by $\gamma$.
Let $\Lambda$ be the logarithmic moment generating function of $D_j$, i.e.,
\begin{equation}  \label{eq: definition lgf}
  \Lambda(\xi)
  = \log \E(e^{\xi D_j})
  = \int_{S^2} |ay+(1-a)z|^\xi \; u(dz).
\end{equation}
As is evident from Figure~\ref{fig:surfaces},
$|ay+(1-a)z|$ is bounded above and below for $a \neq \half$,
and thus then $\Lambda(\xi) < \infty$, for all $\xi \in \R$.
For $a=\half$, using polar coordinates, it can be seen that
$\Lambda(\xi)<\infty$, for $\xi>-1$.
In both cases $\Lambda$ is convex on $\mathbb{R}$,
strictly convex in a neighborhood of $0$, and differentiable
in a neighborhood of $0$.
By general properties of generating functions, we deduce that the Legendre transform of $\Lambda$,
\begin{equation}
  I(x) = \sup_{\xi \in \R} (x\xi - \Lambda(\xi)),
\end{equation}
satisfies  $I(\log \theta)=\min_{x \in \mathbb{R}} I(x)=0$ with $\log \theta= E(D_j)= \Lambda'(0)  <\infty$ and 
 $I$ is finite, continuous and strictly convex in a neighborhood of $\log\theta$. In particular,  $I(x)$ is positive (or $+\infty$) if $x \neq \log \theta$.
By Cram\'er's theorem \cite{MR1619036}, the sum of i.i.d.\ random variables 
\begin{equation}
 \frac{1}{n}\sum_{j=1}^n D_j   =     \frac{1}{n} \left(   \log R_{n} - \log R_{0} \right)
\end{equation}
satisfies the large deviation estimates
\begin{alignat}{2}
  \limsup_{n\to\infty} \frac{1}{n}
  \log \Pr\{\log(R_n)-\log(R_0) \geq xn\} &\leq
  - I(x)
  &\quad& (x \geq \log \theta), \label{eq:log-R-ldp1}\\
  \limsup_{n\to\infty} \frac{1}{n}
  \log \Pr\{\log(R_n)-\log(R_0) \leq xn\} &\leq
  -I(x)
  && (x \leq \log \theta). \label{eq:log-R-ldp2}
\end{alignat}

\begin{proof}[Proof of Proposition~\ref{prop:lln}]
  It suffices to show that, for $x \neq 1$, there exists some $J(x) > 0$ such that, as $s \to \infty$,
  \begin{alignat}{2}
    \Pr\{N_{\theta^{-s}} \geq xs\} &\leq e^{-sJ(x)+o(s)} &\quad& (x > 1), \label{eq:N-ldp1}\\
    \Pr\{N_{\theta^{-s}} \leq xs\} &\leq e^{-sJ(x)+o(s)} && (x < 1). \label{eq:N-ldp2}
  \end{alignat}
  As we show further below, this implies the claim.
  We first observe that $\log \Sigma(R_n) = \log R_n - \log R_0 + \log \Sigma(R_0)$,
  by \eqref{eq:Sigma}.
  Since $\log \Sigma(R_0)= \Sigma(r)$ is a constant, \eqref{eq:log-R-ldp1}--\eqref{eq:log-R-ldp2}
  and continuity of $I(x)$ near $\log \theta$ imply that, as $n\to\infty$,
  \begin{alignat}{2}
    \label{eq:Itilde1}
    \Pr\{ \log \Sigma(R_n) \geq xn\}
    &\leq e^{-nI(x) + o(n)} &\quad& (x > \log \theta),\\
    \label{eq:Itilde2}
    \Pr\{ \log \Sigma(R_n) \leq xn\}
    &\leq e^{-nI(x) + o(n)} && (x < \log \theta).
  \end{alignat}
  Also, by \eqref{eq:Nthetat},
  \begin{align} \label{eq:Nthetar1}
    \{ N_{\theta^{-s}} \geq r \}
    &\subseteq \Big\{ \sum_{j=0}^{[r]-1} \lambda_j \Sigma(R_j)^{-1} \leq \theta^{-s} \Big\}
    ,
    \\ \label{eq:Nthetar2}
    \{ N_{\theta^{-s}} \leq r \}
    &\subseteq \Big\{ \sum_{j=0}^{[r]+1} \lambda_j \Sigma(R_j)^{-1} \geq \theta^{-s} \Big\}
    .
  \end{align}

  To show \eqref{eq:N-ldp1}, let $x > 1$ and choose $x'$ such that $x > x' > 1$.
  By \eqref{eq:Nthetar1}, and using $\Sigma(R_j) \geq 0$, $\lambda_j \geq 0$, and the union bound,
  \begin{align}
    \Pr\{N_{\theta^{-s}} \geq xs\}
    &\leq \Pr\Big\{\sum_{j=0}^{[xs]-1} \frac{\lambda_j}{\Sigma(R_j)} \leq \theta^{-s}\Big\}
    \leq \Pr\Big\{ \frac{\lambda_{[xs]-1}}{\Sigma(R_{[xs]-1})} \leq \theta^{-s} \Big\}
    \nonumber \\
    &\leq \Pr\Big\{ \lambda_{[xs]-1} \leq \theta^{(x'-1)s} \Big\}
      + \Pr\Big\{ \Sigma(R_{[xs]-1}) \geq \theta^{x's} \Big\}
      .
  \end{align}
  Since $\lambda_{[xs]-1}$ is an exponential random variable and $\theta^{x'-1}<1$,
  the first term goes to $0$ exponentially as $s\to\infty$.
  The second term is estimated by
  \begin{equation} \label{eq:It1}
    \Pr\{ \Sigma(R_{[xs]-1}) \geq \theta^{x's} \}
    \leq \Pr\{ \log(\Sigma(R_{[xs]-1})) \geq x's \log\theta \}
    \leq e^{-xs I(\frac{x'}{x} \log \theta) + o(s)}
  \end{equation}
  as $s\to\infty$. From these inequalities, \eqref{eq:N-ldp1} follows, with some $J(x) > 0$.

  To verify \eqref{eq:N-ldp2}, let $x<1$ and choose $x'>0$ such that $x<x'<1$.
  Then, by \eqref{eq:Nthetar2}, since $R_j$ is decreasing,
  and using the union bound,
  \begin{align}
    \Pr\{N_{\theta^{-s}} \leq xs\}
    &\leq \Pr\Big\{\sum_{j=0}^{[xs]+1} \frac{\lambda_j}{\Sigma(R_j)} \geq \theta^{-s}\Big\}
    \leq \Pr\Big\{ \sum_{j=0}^{[xs]+1} \lambda_j \geq \theta^{-s} \Sigma(R_{[xs]+1})\Big\}\nonumber\\
    &\leq \Pr\Big\{\sum_{j=0}^{[xs]+1} \lambda_j \geq \theta^{-(1-x')s}\Big\}
       + \Pr\Big\{ \Sigma(R_{[xs]+1}) \leq \theta^{x's} \Big\} \label{eq:Nldp2-bd1}
       .
  \end{align}
  Since $\E(\sum_{j=0}^{[xs]+1}\lambda_j) = O(xs)$, by Markov's inequality,
  the first term on the right-hand side of \eqref{eq:Nldp2-bd1}
  tends to $0$ exponentially.
  As in \eqref{eq:It1}, the second term on the right-hand side of \eqref{eq:Nldp2-bd1}
  converges to $0$ exponentially since, by \eqref{eq:Itilde2},
  \begin{align}
    \Pr\{\log(\Sigma(R_{[xs]+1})) \leq x's \log\theta \}
    \leq e^{-xs I(\frac{x'}{x} \log \theta) + o(s)}
  \end{align}
  as $s\to\infty$. Thus we conclude \eqref{eq:N-ldp2} for some $J(x) > 0$.
  
  It remains to show that \eqref{eq:N-ldp1}--\eqref{eq:N-ldp2} indeed imply compact convergence of $\Phi_n(s) = N_{\theta^{-ns}}/n \to s$ almost surely.
  For this, we proceed as in the proof of \cite[Proposition~8.2]{MR2664334}.
  Fix $s_0>0$, $\varepsilon >0$.
  Setting $m>2s_0/\varepsilon$, 
  $s_{k} = ks_0/m$,
  since $\Phi_n(s)$ is nondecreasing in $s$,
\begin{align}
  \Pr\Big\{\sup_{s\in[0,s_0]}(\Phi_n(s)-s) \geq \varepsilon \Big\}
  &\leq \sum_{k=1}^{m} \Pr\Big\{\sup_{s\in [s_{k-1},s_{k}]} (\Phi_n(s)-s) \geq \varepsilon\Big\}
  \nonumber\\
  &\leq \sum_{k=1}^{m} \Pr\Big\{\Phi_n(s_k)-s_k \geq \varepsilon - \frac{s_0}{m} \geq \half \varepsilon\Big\}
  \nonumber\\
  &\leq \sum_{k=1}^{m} \Pr\Big\{N_{\theta^{-ns_{k}}} \geq n s_{k}(1 + \tfrac12 \varepsilon m/s_0) \Big\}
\end{align}
which is bounded by $c^{-1} e^{-c n}\to 0$ as $n\to\infty$, for some $c=c_{\varepsilon,s_0}>0$.
Analogously,
\begin{equation}
  \Pr\Big\{\sup_{s\in[0,s_0]} (s-\Phi_n(s)) \geq \varepsilon\Big\} \leq c^{-1} e^{-c n}
 .
\end{equation}
Since the right-hand sides are summable,
the Borel--Cantelli Lemma immediately implies that $\sup_{s\in[0,t_0]} |\Phi_n(s) -s| \to 1$ with probability 1,
as claimed.
\end{proof}

\subsection{Proof of Proposition~\ref{prop:clt}}
\label{sec:pfclt}

The proof of Proposition~\ref{prop:clt} follows
the well-known route to prove convergence to Brownian motion by martingale approximation; see, e.g., the introduction of \cite{MR834478}.
We recall that $(Y_n)$ is a Markov chain on $S^2$ and that $(\lambda_n)$ is a sequence of i.i.d.\ unit exponential random variables, independent of $(Y_n)$.
We find it convenient to consider the pair $(Y_n,\lambda_n)$ as a Markov chain on $S^2 \times (0,\infty)$, and denote its transition operator by $P$ and
its kernel by $p(y,l; dy',dl'), \; y,y'\in S^2, l,l' > 0$.
When the exponential random variables are irrelevant, we also write transition kernel of $(Y_n)$
as $p(y; dy')$.
Then, with $b$ as in \eqref{eq:constU},
\begin{equation} \label{eq:bprop}
  \int l'y' \; p(y, l; dy', dl')
  =\int y' \; p(y; dy')
  = by.
\end{equation}
Define $h, V: S^2 \times \R_+ \to \R^3$ by
\begin{equation} 
  h(y,l) = ly + \frac{b}{1-b}y, \quad
  V(y,l) = ly, \quad (y \in S^{2}, l > 0)
\end{equation}
so that $V = (1-P)h$, by \eqref{eq:constU}. Define
\begin{equation} \label{eq:Yprimedef}
  Y_n' = h(Y_n,\lambda_n) - Ph(Y_{n-1},\lambda_{n-1})
  .
\end{equation}
Then $\E(Y_{n+1}' | \Sigma_n) = 0$, i.e., $(Y_{n}',\Sigma_n)$ is a martingale difference sequence.

\begin{lem} \label{lem:cov}
  \begin{equation} \label{eq:covconv}
    \frac{1}{n} \sum_{j=1}^n \E(Y_j'^\alpha Y_j'^\beta | \Sigma_{j-1}) \to \frac{2\delta_{\alpha\beta}}{3(1-b)}
    \quad (\text{in probability}).
  \end{equation}
\end{lem}

\begin{proof}
  Let $q_{\alpha\beta}(y) = \E(Y_j'^\alpha Y_j'^\beta | Y_{j-1}=y) = \E(Y_j'^\alpha Y_j'^\beta | Y_{j-1}=y, \lambda_{j-1} = l)$
  with arbitrary $l$.
  In a short calculation, we first show that
  \begin{equation}
    C_{\alpha\beta} := \E(q_{\alpha\beta}(U)) = \frac{2\delta_{\alpha\beta}}{3(1-b)},
  \end{equation}
  where $U$ is a uniform random variable on $S^2$. Indeed, by \eqref{eq:bprop}--\eqref{eq:Yprimedef},
  \begin{equation}
    Y_j'
    = (\lambda_j + \frac{b}{1-b}) Y_j - \frac{b}{1-b} Y_{j-1}
    .
  \end{equation}
  Since $\lambda_j$ is a unit exponential random variable,
  \begin{equation}
    \E(\lambda_j + \frac{b}{1-b}) = 1+\frac{b}{1-b} = \frac{1}{1-b}
    ,\quad
    \E((\lambda_j + \frac{b}{1-b})^2)
    = \frac{1+(1-b)^2}{(1-b)^2}
    .
  \end{equation}
  Since $\E(Y_j^\alpha|\Sigma_{j-1})Y_{j-1}^\beta  =  b Y_{j-1}^\alpha Y_{j-1}^\beta$ by \eqref{eq:bprop},
  it follows that
  \begin{equation} \label{eq:qbd}
    \E(Y_j'^\alpha Y_j'^\beta | \Sigma_{j-1})
    = \frac{1}{(1-b)^2} ( (1+(1-b)^2) \E(Y_j^\alpha Y_j^\beta | \Sigma_{j-1}) - b^2 Y_{j-1}^\alpha Y_{j-1}^\beta)
    ,
  \end{equation}
  and therefore, since $\E(U^\alpha U^\beta) = \frac{1}{3}\delta_{\alpha\beta}$,
  \begin{equation}
    C_{\alpha\beta}
    = \frac{1}{(1-b)^2} ( (1+(1-b)^2-b^2) \E(U_j^\alpha U_j^\beta)
    = \frac{2\delta_{\alpha\beta}}{3(1-b)} 
    ,
  \end{equation}
  as claimed.
  Let $p^n(y,dz)$ be the $n$-step transition probability of the Markov chain $(Y_n)$,
  defined from the one-step transition probability $p(y,dz)$ by the Chapman-Kolmogorov equations.
  From the definition of $p(y,dz)$ above \eqref{eq:bprop}, it can be seen (cf.\ Figure~\ref{fig:surfaces}) that
  there exists $n > 0$ such that 
  the density of the absolutely continuous component,
  with respect to the uniform measure on $S^2$, of $p^n(y, \cdot)$
   is bounded below uniformly by a strictly positive constant.
  Thus the Markov chain $Y_j$  satisfies Doeblin's condition \cite[p.197, condition (D')]{MR1038526}.
  It follows that, if $\bar Y_j$ is an independent instance of this Markov chain
  with initial
  stationary distribution $u$, the total variation distance between the
  distributions of $Y_j$ and $\bar Y_j$ tends to $0$, exponentially fast.
  Since the sequence $\bar Y_j$ is stationary and ergodic,
  $\frac{1}{n}\sum_{j=0}^{n-1} q(\bar Y_j) \to q(U) = C$ almost surely, by the ergodic theorem,
  and the proof follows.
\end{proof}

\begin{proof}[Proof of Proposition~\ref{prop:clt}]
  Given Lemma~\ref{lem:cov},
  the proof is a straightforward consequence of a standard functional central
  limit theorem for martingales \cite[Theorem~3.3]{MR668684}.
  For the reader's convenience, we restate a special case of it
  as Theorem~\ref{thm:martingaleclt}.
  Let
  \begin{equation} \label{eq:Znprime}
    Z_n' := \sum_{j=0}^n Y_j' = \sum_{j=1}^n V(Y_j,\lambda_j) - Ph(Y_0,\lambda_0) + Ph(Y_n,\lambda_n).
  \end{equation}
  Since
  \begin{equation}
    |Y_j'|
    = \left|(\lambda_j + \frac{b}{1-b}) Y_j - \frac{b}{1-b} Y_{j-1}\right| \leq \lambda_j + \frac{2b}{1-b},
  \end{equation}
  it follows that
  \begin{equation}
    \E(Y_j'^\alpha Y_j'^\beta 1_{|Y_j'| \geq \sqrt{n} \varepsilon}) \leq \Pr\{\lambda_j \geq \sqrt{n}\varepsilon - O(1)\} \to 0.
  \end{equation}
  Together with the fact that $(Y_j',\Sigma_j)$ is a martingale difference sequence, all conditions of
  Theorem~\ref{thm:martingaleclt} are satisfied for $M_{j} = Y_j'$.
  Setting  $\tilde \sigma = \frac{2}{3}(1-b)^{-1}$, it follows that
  $(\frac{1}{\sqrt{n}}Z_{[ns]}') \to (\tilde\sigma B_s)$ in the Skorohod space $D([0,s_0],\R^3)$,
  for any $s_0>0$. Since
  \begin{equation}
    \sum_{j=0}^n \lambda_jY_j = \sum_{j=1}^n V(Y_j,\lambda_j),
  \end{equation}
  the proof is concluded by the observation that the two ``boundary terms'' in \eqref{eq:Znprime} are negligible in the limit
  after multiplication by $\frac{1}{\sqrt{n}}$.
  This is trivial since
  \begin{equation}
    |Ph(Y_0,\lambda_0) + Ph(Y_n,\lambda_n)|
    =
    \left|\frac{b}{1-b}Y_0 - \frac{b}{1-b}Y_n\right|
    \leq \frac{2b}{1-b},
  \end{equation}
  so the boundary terms are in fact bounded.
\end{proof}

\subsection{Proof of Theorem~\ref{thm:position}}
\label{sec:pfpos}

\begin{proof}[Proof of \eqref{eq:Kasymp}]
By Proposition~\ref{prop:lln}, $N_{\theta^{-s}}/s \to 1$ as $s \to \infty$ a.s.,
i.e., $N_t= (\log t/(-\log \theta))(1+o(1))$.
Moreover, by \eqref{eq:log-R-ldp1}--\eqref{eq:log-R-ldp2},
$\log |K_n|-\log |K_0| = n \log \theta(1+o(1))$ almost surely. It follows that
\begin{equation}
  \log |K_t|-\log |K_0| =
  \frac{-\log t}{N_t\log \theta} (\log |K_{N_t}|-\log |K_0|)(1+o(1)) = -\log t(1+o(1))
\end{equation}
almost surely, as claimed.
\end{proof}

\begin{proof}[Proof of \eqref{eq:XEinf}]
We start from \eqref{eq:Xtsum}.  Since $(t-T_{N_t}) |K_{N_t}| \leq \lambda_{N_t}$
is finite almost surely, and using again the bounds on the random time change
(Proposition~\ref{prop:lln}), it suffices to prove
\begin{equation}
  \sup_n \Big| \E\Big( \sum_{j=1}^n \lambda_j Y_j \Big) \Big|   <\infty, \qquad \text{and} \qquad  \E \Big|\sum_{j=1}^n \lambda_j Y_j \Big|    \to \infty.
\end{equation}
The first claim follows from the fact that the distribution of  $\lambda_j Y_j$ converges
to that of $\lambda U$ exponentially fast, in total variation distance
(see proof of Lemma~\ref{lem:cov}),
and $\E(\lambda U)=\E(U)=0$. 
To verify the second claim, we apply  Markov's inequality,
\begin{equation}
 \E \Big|\sum_{j=1}^n \lambda_j Y_j \Big|   \geq  \sqrt{n}\,  \Pr\Big\{ \frac{1}{\sqrt{n}} \Big|\sum_{j=1}^n \lambda_j Y_j \Big| \geq  1 \Big\}
\end{equation}
and can then invoke Proposition \ref{prop:clt}  to argue that  $\Pr\{\cdot\}$
on the right hand side is bounded away from $0$.
\end{proof}

\section*{Acknowledgement}

Most of the research that has led to this paper was carried out while
the three authors were at the Institut f\"ur Theoretische Physik, ETH
Z\"urich.
The work of R.B. has also been supported by the National Science
Foundation under grant No.\ DMS-1128155.  This paper was completed
during a stay of R.B. and J.F. at the Institute for Advanced Study in
Princeton. The stay of J.F. at IAS has been supported by `The Fund for
Math' and `The Robert and Luisa Fernholz Visiting Professorship Fund'.

\appendix

\section{Martingale functional central limit theorem}
\label{app:martingaleclt}

The following theorem is a special case of \cite[Theorem~3.3]{MR668684}.

\begin{thm} \label{thm:martingaleclt}
  Let $(M_n, \Sigma_n)$ be an $\R^d$-valued martingale difference sequence, i.e., $(\Sigma_n)$
  is a filtration of $\sigma$-algebras and 
  $E(M^\alpha_{n+1} | \Sigma_n) = 0$ for $1\leq \alpha \leq d$.
  Let
  \begin{equation}
    Z_n(s) = \frac{1}{\sqrt{n}}
    \sum_{j=1}^{[ns]} M_j .
  \end{equation}
  Assume there is $\sigma^2 >0$ such that for all $t > 0$,
  $1\leq \alpha, \beta \leq d$, $\varepsilon > 0$,
  \begin{gather}
    \frac{1}{ns} \sum_{j=1}^{[ns]} \E(M_j^\alpha M_j^\beta | \Sigma_{j-1}) \to \sigma^2 \delta_{\alpha\beta},
    \\
    \frac{1}{ns} \sum_{j=1}^{[ns]} \E(M_j^\alpha M_j^\beta 1_{|M_j| \geq \sqrt{n}\varepsilon} | \Sigma_{j-1}) \to 0
    ,
    \quad (n\to\infty),
  \end{gather}
  in probability  Then $(Z_n(s))_s \to (\sigma B(s))_s$
  in distribution in $D([0,s_0],\R^d)$ for all $s_0 > 0$,
  where $D$ is equipped with the Skorohod topology.
\end{thm}

\section{Proof of Proposition~\ref{prop:markovboltzmann}}
\label{sec:pfmarkovboltzmann}

\begin{proof}
  The strong Markov property for $(X_t, K_t)$ implies that, for any stopping time $S$,
  \begin{equation}
    \Gamma_t f(x,k)
    = \E^{x,k}(f(X_t, K_t))
    = \E^{x,k}(\Gamma_{t-S} f(X_S,K_S)).
  \end{equation}
  We start the semigroup at the time of the first jump, or rather at $S = \min \{t, T_1\}$:
  \begin{equation}
    \Gamma_t f(x,k)
    = \E^{x,k}(f(X_t, K_t);\; T_1 > t)  + \E^{x,k}((\Gamma_{t-{T_1}}f)(X_{T_1}, K_{T_1}) ;\; T_1 \leq t)
  \end{equation}
  In the first term, $X_t = x+\frac{k}{m} t$ and $K_t=k$, and hence
  \begin{equation}
    \E^{x,k}(f(X_t,K_t);\;T_1>t)
    = f(x+\frac{k}{m}t, k) \Pr^{x,k}(T_1 > t)
    = f(x+\frac{k}{m}t, k) e^{-\Sigma(k)t}.
  \end{equation}
  In the second term, since $T_1$ is exponentially distributed with parameter $\Sigma(k)$,
  \begin{equation}
    \E^{x,k}(\Gamma_{t-S}f(X_S,K_S) ;\; T_1 \leq t)
    = \int_0^t \left( \int p(k, dk') (\Gamma_{t-\tau}f)(x+\frac{k}{m}\tau, k') \right) e^{-\Sigma(k)\tau} \Sigma(k) \; d{\tau}.
  \end{equation}
  Using $q(k,dk') = \Sigma(k) p(k,dk')$, and substituting $r=t-\tau$, this is
  \begin{equation}
    \E^{x,k}(\Gamma_{t-S}f(X_S, K_S) ;\; T_1 \leq t)
    = \int_0^t \left( \int q(k, dk') (\Gamma_{r}f)(x+\frac{k}{m}(t-r), k') \right) e^{-\Sigma(k)(t-r)} \; dr.
  \end{equation}
  Putting both terms together, we obtain
  \begin{multline}
    \Gamma_tf(x,k) 
    = f(x+\frac{k}{m}t,k) e^{-\Sigma(k)t} 
    \\+ \int_0^t \left( \int q(k, dk') (\Gamma_{r}f)(x+\frac{k}{m}(t-r)),k') \right) e^{-\Sigma(k)(t-r)} \; dr .
  \end{multline}
  The last expression is just the Duhamel formula for the semigroup generated by $L = A - \Sigma + K$,
  where $A=\frac{k}{m} \nabla$ is the advection operator, and
  considering the gain operator $K$ with kernel $q$ as a perturbation to the semigroup generated $A-\Sigma$.
  Hence, the claim is verified.
\end{proof}


\end{document}